\documentclass[reqno,12pt]{article}
\usepackage{amssymb,amsthm,amsmath,amsfonts}
\usepackage{epsfig}

\theoremstyle{plain}
\begingroup

\newtheorem{thm}{Theorem}[section]
\newtheorem{theorem}{Theorem}[section]

\newtheorem{lemma}[thm]{Lemma}

\endgroup
\theoremstyle{definition}
\newtheorem{defn}{Definition}[section]
\newtheorem{definition}{Definition}[section]

\newtheorem{remark}[defn]{Remark}

\newtheorem{example}[defn]{Example}

\theoremstyle{remark}

\numberwithin{equation}{section}
\numberwithin{figure}{section}

\DeclareMathOperator{\re}{Re}

\def\I{\mathrm{i}}

\def\D{{\mathbb D}}
\def\R{{\mathbb R}}
\def\C{{\mathbb C}}

\begin{document}

\title{
Vortex motion and geometric function theory: the role of connections}
\author{
Bj\"orn Gustafsson\textsuperscript{1},
}

\date{November 23, 2018}

\maketitle


\begin{abstract} 
We formulate the equations for point vortex dynamics on a closed two dimensional Riemann manifold in the language
of affine and other kinds of connections. The speed of a vortex is then expressed in terms of the difference between an affine connection derived from
the coordinate Robin function and the Levi-Civita connection associated to the Riemannian metric. 

A Hamiltonian formulation of the same dynamics is also given. The relevant  Hamiltonian function consists of two main terms. One of the terms is the well-known
quadratic form based on a matrix whose entries are Green and Robin functions, while the other term describes the energy contribution from those
circulating flows which are not implicit in the Green functions.  One main issue of the paper is the clarification of the somewhat intricate exchanges of energy 
between these two terms of the Hamiltonian. 

\end{abstract}



 \footnotetext[1]
{Department of Mathematics, KTH, 100 44, Stockholm, Sweden.\\
Email: \tt{gbjorn@kth.se}}

\noindent {\it Acknowledgements:} I am grateful to Stefanella Boatto for discussions, support, encouragement, and
for inspiring papers (together with collaborators) which have made me return to the field of vortex dynamics after many years of absence.


\section{Introduction}

In this paper we study point vortex motion on closed two dimensional Riemannian manifolds, mainly from the
point of view of Riemann surfaces with additional structure. 
This is a quite classical subject, even though equations describing the dynamics of vortices 
in general situations have been made explicit rather recently
(see for example \cite{Hally-1980, Boatto-Koiller-2013, Dritschel-Boatto-2015, Grotta-Ragazzo-Barros-Viglioni-2017}). 
Our studies however have some new features, to be described below.

Since point vortices have infinite energy some kind of renormalization is necessary when passing from the general Euler equations
for the fluid, assumed non-viscous and incompressible, to the dynamics for the vortices, often written in Hamiltonian form. This
renormalization usually involves some kind of Robin function, defined as the first regular term (the constant term) in the expansion of the 
the appropriate Green function at the singular point (the vortex).  

The singular part of the Green function contains the logarithm of the
distance to the pole. Usually this distance is interpreted as the geodesic distance in terms of the metric. However there is also the
possibility to work within a local holomorphic chart and use the Euclidean distance in the range of the coordinate chosen.
That gives a Robin function which is not really a function (because it depends on the choice of coordinate), instead it is a kind of
connection, in this paper called $0$-connection. There are also $1$-connections, which are the same as affine connections in ordinary
differential geometry, and these are even more relevant for vortex motion because they directly give the speed of the vortex, as
measured in a local coordinate. Finally there are $2$-connections, or projective, or Schwarzian, connections. These may play 
a role in describing the motion of vortex pairs, or vortex dipoles, but at present that remains an open question.

We may remind that various kinds connections nowadays play important roles in  mathematical physics, 
like the electromagnetic and Yang-Mills connections in quantum field theory, and the role the stress-energy tensor as being a
Schwarzian connection in conformal field theory and string theory. Therefore a treatment of vortex motion from the point of
view of connections might be of some interest, even though strictly speaking no new results have so far been produced this way.

A second new feature of the paper is that we are careful of treating Riemann surfaces of arbitrary genus, keeping full respect for the 
Kelvin law of conservation of circulations. This law enters the vortex dynamics in a way which cannot been ignored.
For example, in a Hamiltonian description of the dynamics, the Hamiltonian function, which represents the renormalized energy,
consists of two major parts. The first part is the usual quadratic form described by a matrix which has Green and Robin functions
as its entries. 
The second part is a quadratic form which gives the energy of circulating  flows beyond those which are already
present in the Green function. The problem is that there is a nontrivial interaction between the two parts of the Hamiltonian, due to the fact that
the electrostatic type Green function we are working with does not respect the Kelvin law. As the poles (vortices) move
the global circulations implicit in the Green function change, and that has to be compensated by the second term in the Hamiltonian. 

One could imagine replacing the electrostatic Green function by a hydrodynamic Green function, like in the multiply connected planar case.
Effectively this is actually what we do, in a certain sense, but in the Riemann surface case this function will not be single-valued (if genus is strictly positive).
In the same vein, our stream functions will generally not be single-valued. 

The main specific results in the paper are Theorems~\ref{thm:dzkdt} and \ref{thm:hamiltoneq}. The first of these describes the dynamics for a system
of point vortices in terms of the difference between two affine connections, one coming from the metric (then being the same as the classical
Levi-Civita connection) and the other being a connection canonically derived from the Robin function. 
This difference is a differential $1$-form, and it is the same as what is obtained when treating the Robin function in the traditional way.
However, our affine Robin connection contains the information of the compensating circulating flows alluded to above, hence the
result might be new anyway.

The derivation of Theorems~\ref{thm:dzkdt} goes via a weak formulation of Euler's equation, Lemma~\ref{lem:eulerweak}, 
which may look a little strange, and {\it ad hoc}, but it can be motivated on physical grounds.
The ideas behind it go back to some previous unpublished work \cite{Gustafsson-1979, Flucher-Gustafsson-1997}. 

In Theorem~\ref{thm:hamiltoneq} we confirm the vortex dynamics obtained in Theorems~\ref{thm:dzkdt} by an independent derivation
based on renormalization of the kinetic energy of the flow. This gives the appropriate Hamiltonian function, and together with the natural 
symplectic structure on the phase space the dynamics follows automatically.

The paper is much inspired by general treatments in physics and differential geometry such as Flanders \cite{Flanders-1963}, 
Frankel \cite{Frankel-2012}, Arnold and Khesin \cite{Arnold-1978, Arnold-Khesin-1998}, Schutz \cite{Schutz-1980}. 
In particular, we take the liberty to start the paper by giving the general equation in non-viscous fluid dynamics in a quite general form, following  
\cite{Schutz-1980, Frankel-2012}. This is based on Lie derivatives, which is a natural tool in fluid dynamics, 
but after the quite brief general discussion in Section~\ref{sec:general} our treatment is mostly based on differential forms. 

As for the more specific questions of point vortex dynamics our paper owes much to several recent articles. Here
the papers \cite{Boatto-Koiller-2008, Boatto-Koiller-2013} by Boatto and Koiller deserve special mention because of their 
grand perspectives and historical accounts. In addition, \cite{Dritschel-Boatto-2015, Grotta-2017, Grotta-Ragazzo-Barros-Viglioni-2017} and  (the less recent)
 \cite{Grotta-Koiller-Oliva-1994} are highly relevant. Interesting historical reviews and general treatments can also be found in 
\cite{Llewellyn-2011}. Some other important sources are \cite{Hally-1980, Marsden-Weinstein-1983, Steiner-2003, Steiner-2005, Sakajo-Shimizu-2016}.

As for the further organization of the paper, we turn in Section~\ref{sec:Hodge} to geometric function theory (closed Riemann surfaces),
defining first various kinds of Green potentials and harmonic differentials via Hodge theory. Our way of introducing the Green function is slightly different from that
in papers mentioned above, but the final result is of course the same. The one point Green function with compensating uniform background sink comes out 
directly from Hodge decomposition.  Using the fundamental two point Green potential (which does not depend on the metric) we construct the canonical
harmonic differentials. We expose the full details here, despite the theory being very classical (already in  \cite{Weyl-1964}, for example), because we
need the explicit links between the fundamental potentials and differentials on the Riemann surface in order to faithfully implement 
the Kelvin law for circulations.

We proceed in Section~\ref{sec:connections} to define affine and projective connections. In Section~\ref{sec:singular parts} we show how 
leading terms in regular parts of Green functions and standard differentials can be interpreted as such connections. 
This also reveals the philosophy behind connections entering into vortex dynamics. 
Actually connections have been hidden for a long time in fluid dynamics, being implicit in transformation laws for various quantities.
One example is the inhomogenous transformation law for Routh's stream function, which appears already in Lin's classical work \cite{Lin-1941a, Lin-1941b, Lin-1943}.

In Section~\ref{sec:euler} we return to fluid dynamics by giving weak forms, suitable for our purposes, of the Euler equation, and in 
Section~\ref{sec:point vortices} we fuse this with the theory of connections, to obtain the point vortex dynamics. The result
obtained, Theorems~\ref{thm:dzkdt}, is still a little vague in the sense that the appearance of circulations, beyond those present in
the Green function, is somewhat implicit. This we take care of in Section~\ref{sec:hamiltonian}, where we make the Hamiltonian function
fully explicit, and in particular we clarify how the mentioned extra circulations are linked to the locations of the vortices. 
In Section~\ref{sec:hamiltoneq} finally we prove that the vortex dynamics given by Hamilton's equations agrees with that given 
previously (in Theorems~\ref{thm:dzkdt}).


\section{Fluid dynamics on Riemannian manifolds}\label{sec:general}

\subsection{General notations and assumptions}

We consider the dynamics of a nonviscous incompressible fluid on a compact Riemannian manifold
of dimension two. We start by setting up the basic equations in somewhat more generality, in arbitrary dimension, 
following treatments based on using the Lie derivative, as exposed for example in
Frankel \cite{Frankel-2012} and Schutz \cite{Schutz-1980}. These sources also provide the standard notations in differential geometry to be used.
Other treatises in fluid dynamics,  suitable for our purposes, are \cite{Marchioro-Pulvirenti-1994, Arnold-Khesin-1998, Newton-2001}.

Let ${M}$ be a Riemannian manifold of dimension $n$. We consider the flow of a nonviscous  fluid on $M$. Some general notations:
\begin{itemize}

\item $\rho =\rho (x,t)=\text{density} >0$

\item $p=p(x,t)=\text{pressure}$

\item $ds^2= g_{ij} dx^idx^j \quad \text{Riemannian metric (Einstein summation convention in force)}$

\item ${\bf v}={\bf v}(x,t)= v^j \frac{\partial}{\partial x^j} \quad \text{fluid velocity vector field}$

\item ${\bf \nu}=  v_i {d x^i}=g_{ij} v^j dx^i  \quad \text{corresponding one-form (using the metric)}$

\item $\omega = d\nu\quad \text{vorticity field, as a two-form}$

\item $*\omega =  \text{ Hodge star acting on a differential form } \omega$

\item $*1 = {\rm vol}^n, \text{ the volume form on } M$

\item $i_{\bf v} =i({\bf v})=\text{interior multiplication with respect to a vector field }{\bf v}$

\item $\mathcal{L}_{\bf{v}}= \text{ Lie derivative  with respect to a vector field } {\bf v}$

\end{itemize}

The Hodge star and interior multiplication are related by
\begin{equation}\label{Hodgei}
i_{\bf v} {\rm vol}^n=*\nu,
\end{equation}
where ${\bf v}$ and $\nu$ are linked via the metric tensor.
The Lie derivative $\mathcal{L}_{\bf{v}}$ can act on all kinds of tensors. When it acts on differential forms the homotopy formula (or the (Henri) Cartan formula) 
\begin{equation}\label{cartan}
\mathcal{L}_{\bf{v}}=d\circ i_{\bf{v}}+  i_{\bf{v}}\circ d
\end{equation}
is very useful. It shows in addition that $\mathcal{L}_{\bf{v}}$ commutes with $d$.


\subsection{Basic equations}

We collect here the basic equations of fluid dynamics in our context. See \cite{Frankel-2012, Schutz-1980} for details.

\begin{theorem}
\label{thm:dynamics}

The basic equations for a nonviscous fluid in $n\geq 2$ dimensions are

\begin{itemize}

\item[(i)] Euler' equation (conservation of momentum):
\begin{equation}\label{euler}
(\frac{\partial }{\partial t}+\mathcal{L}_{\bf v})({\bf \nu})=d(\frac{1}{2}|{\bf v}|^2)-\frac{dp}{\rho}
\end{equation}
\item[(ii)] Continuity equation (conservation of mass):
\begin{equation}\label{continuityeq}
(\frac{\partial}{\partial t}+\mathcal{L}_{\bf{v}})(\rho \cdot {\rm vol}^n)=0
\end{equation}
\item[(iii)] Adiabatic flow (conservation of entropy, adiabatic case):
$$
(\frac{\partial}{\partial t}+\mathcal{L}_{\bf{v}})({\rm entropy})=0
$$
\item[(iv)] A constitutive  relation between $\rho$ and $p$, for example of the form
\begin{equation}\label{constitutive}
\rho=\rho(p).
\end{equation}
\end{itemize}
\end{theorem}

As for Euler's equation, the standard vector version of it is
$$
\frac{\partial {\bf v}}{\partial t}+({\bf{v}}\cdot\nabla) {\bf v}=-\frac{dp}{\rho},
$$
and rewriting this as an equation for differential forms gives 
\begin{equation}\label{euler3}
\frac{\partial \nu}{\partial t}+i_{\bf v}(d\nu)=-\frac{1}{2}d(i_{\bf v}(\nu))-\frac{dp}{\rho},
\end{equation}
or
$$
\frac{\partial \nu}{\partial t}+d(i_{\bf v}(\nu))+i_{\bf v}(d\nu)=+\frac{1}{2}d(i_{\bf v}(\nu))-\frac{dp}{\rho}.
$$
The latter is the same as (\ref{euler}).

When (\ref{constitutive}) holds the form $dp/\rho$ is exact, namely equal to the exterior derivative of a primitive function of $1/\rho(p)$.
This means that the pressure $p$ becomes a redundant variable:
Euler's equation (\ref{euler}) simply becomes the statement that $(\frac{\partial }{\partial t}+\mathcal{L}_{\bf v})({\bf \nu})$ is exact, say 
\begin{equation}\label{exact}
(\frac{\partial }{\partial t}+\mathcal{L}_{\bf v})({\bf \nu})=d\phi.
\end{equation}
And afterwords $p$ can be recovered, up to an additive constant, from
\begin{equation*}\label{exact2}
\phi= \frac{1}{2}|{\bf v}|^2-\int \frac{dp}{\rho(p)}.
\end{equation*}


On acting by $d$ (exterior derivative) on Euler's equation and using the constitutive relation,
the laws of conservation of vorticity, local and global follow immediately (recall that $\omega=d\nu$):

\begin{theorem}[Helmholtz, Kirchhoff, Kelvin]\label{thm:helmholtz}
\begin{equation}\label{helmholtz}
(\frac{\partial}{\partial t}+\mathcal{L}_{\bf{v}})(\omega)=0,
\end{equation}
and 
\begin{equation}\label{ddtgamma}
\frac{d}{dt} \oint_{\gamma(t)} \nu =0 
\end{equation}
for any closed curve $\gamma(t)$ which moves with the fluid.

Conversely, if (\ref{ddtgamma}) holds for some time dependent form $\nu$ and for every closed $\gamma$
moving with the flow of the corresponding vector field ${\bf v}$, then (\ref{exact}) holds, and with (\ref{constitutive})
in force (\ref{euler}) can be recovered.  
\end{theorem}

Following \cite{Marchioro-Pulvirenti-1994} we shall refer to (\ref{ddtgamma}) as Kelvin's law for circulations.

\begin{proof}
For (\ref{helmholtz}) one uses that $d$ commutes with $\frac{\partial}{\partial t}$
and with $\mathcal{L}_{\bf{v}}$, and for (\ref{ddtgamma}) that ``$\gamma(t)$ moving with the fluid'',
together with (\ref{exact}), gives that 
$$
\frac{d}{dt} \oint_{\gamma(t)} \nu = \oint_{\gamma(t)} (\mathcal{L}_{\bf{v}}\nu +\frac{\partial \nu}{\partial t})
=\oint_{\gamma(t)}d\phi=0.
$$ 
For the converse statement one simply runs the above arguments backwards.

\end{proof}

The statement that (\ref{ddtgamma}) holds for all $\gamma(t)$ as stated can be thought of as a weak formulation
of the Euler equation. Helmholtz equation (\ref{helmholtz}) alone does not imply Euler's equation (\ref{euler})
if $M$ has nontrivial topology.

More ambitious than the above would be to treat fluid mechanics from the point of view of general relativity, and to
include all kinds of thermodynamic quantities, like what is done in \cite{Schutz-1970}. However, we shall rather go in the 
other direction and simplify to the case of incompressible fluids in two dimensions. 
Eventually we shall specialize to the case of having the vorticity concentrated to finitely many points. 
It follows from (\ref{helmholtz}), (\ref{ddtgamma}) that such a situation is preserved in time if it holds initially, hence the system will 
look like a classical Hamiltonian system having
only finitely many degrees of freedom. From a Riemann surface point of view it becomes a dynamical system for Abelian differentials 
of the third kind (classical terminology, see \cite{Weyl-1964, Farkas-Kra-1992}). Such a differential starts moving once the 
Riemann surface has been provided with a Riemannian metric.


\section{Hodge theory and Green functions}\label{sec:Hodge}


\subsection{Some rudimentary Hodge theory}\label{sec:hodge}

Green functions and differentials with prescribed singularities and periods can most conveniently be introduced via Hodge theory.
We briefly recall the needed concepts here, confining ourselves to the case of two dimensions. For more details  and general notational conventions, see
\cite{Frankel-2012, Warner-1983}. 

We assume that ${M}$ is a compact (closed) oriented Riemannian manifold of dimension $n=2$. We may think of $M$ as a Riemann surface provided with
a Riemannian metric which is compatible with the conformal structure, namely on the form
$$
ds^2=\lambda(z)^2(dx^2+dy^2)=\lambda(z)^2 |dz|^2.
$$
Here $z=x+\I y$ is any local holomorphic coordinate on $M$, and $\lambda>0$ is assumed to be a smooth.

The Laplacian operator acting on a form $\omega$ of any degree
$p=0,1,2$ is $\Delta \omega = *d*d \omega+d*d*\omega$, where the star is the Hodge star operator.  When acting on $1$-forms the Hodge star does not depend on the 
metric (only on the conformal structure), in fact $*dx=dy$, $*dy=-dx$. On $0$-forms it does depend on the metric, for example
$*1=\lambda(z)^2dxdy= {\rm vol}^2$, the volume form on $M$ (invariant area). Similarly for $2$-forms: $*{\rm vol}^2=1$.

A $p$-form $\omega$ is {\it harmonic} if $\Delta \omega=0$. The
only $0$-forms (functions) which are harmonic on all of $M$ are the constant functions, and the only global harmonic $2$-forms are the constant multiples of the 
volume form. The space of harmonic $1$-forms on $M$ has dimension $2\texttt{g}$, where $\texttt{g}$ is the genus of $M$, and by the Hodge-deRham theory
this space is dual in a natural way to the first homology group of $M$. Similarly for harmonic $0$- and $2$-forms. 
Having $\Delta \omega=0$ on all of $M$ is equivalent to that the two equations  $d\omega=0$, $d* \omega=0$ hold. 

The natural inner product on the space of $p$-forms is
\begin{equation}\label{innerproduct}
(\omega_1, \omega_2)_p=\int_M \omega_1\wedge *\omega_2.
\end{equation}
We denote by $L^2(M)_p$ the  Hilbert space of $p$-forms with $(\omega,\omega)_p<\infty$ and inner product (\ref{innerproduct}).
The Hodge theorem \cite{Warner-1983} says that any  $\omega\in L^2(M)_p$
has an orthogonal decomposition 
\begin{equation}\label{hodge}
\omega= \eta +d\mu+\delta  \nu,
\end{equation}
where $\delta=-*d*$ is the coexterior derivative, $\eta$ is a harmonic $p$-form, $\mu$ is a coexact 
$(p-1)$-form,  and $\nu$ is an exact $(p+1)$-form. 
In (\ref{hodge}), the forms $\mu$ and $\nu$ are not uniquely determined, only $d\mu$, $\delta\nu$, and $\eta$ are.
The decomposition can however be made more precise as
\begin{equation}\label{hodge1}
\omega= \eta +d\delta\tau+\delta d  \tau=\eta- \Delta \tau,
\end{equation}
where the $p$-form $\tau$ becomes unique on requiring that it shall be orthogonal to all harmonic forms. 

We shall need the Hodge theorem only for $2$-forms, in which case it makes the one point Green function appear naturally. 
Indeed, if $\omega$ is a $2$-form then one term in the Hodge decomposition disappears immediately, and the rest can be written
\begin{equation}\label{Hodgeomega}
\omega= c_{\omega}\, {\rm vol}^2 -d*dG^\omega.
\end{equation}
Here the first term represents the harmonic $2$-form in the decomposition, with the constant $c_{\omega}$ necessarily being the mean value of $\omega$
with respect to volume:
\begin{equation}\label{tvol}
c_{\omega} = \frac{1}{{\rm vol}^2(M)}\int_M\omega. 
\end{equation}
The second term represents what is in the range of $d$, and the function $G^\omega$ there is to be normalized (as for the
additive level) so that it is orthogonal to all harmonic $2$-forms, i.e., so that
\begin{equation}\label{normalizationG}
\int_M G^\omega \,{\rm vol}^2 =0.
\end{equation}
We then call $G^\omega$ the  {\it Green potential} of $\omega$. 
We see that $-d*dG^\omega=\omega$ provided $\omega$ has net mass zero, otherwise there will be a compensating countermass (multiple of ${\rm vol}^2$), 
namely the first term  in the right member of (\ref{Hodgeomega}).

With $\nu$  the fluid velocity $1$-form, the inner product $(\nu, \nu)_1$ has the interpretation of being the kinetic energy
of the flow. For a function (potential) $u$ we consider the Dirichlet integral $(du,du)_1$ to be its energy.
Thus constant functions have no energy. Similarly, for $2$-forms we consider 
${\rm vol}^2$ to have no energy, and the energy  of $\omega$ is then
defined to be the energy $(dG^\omega, dG^\omega)_1$ of its Green potential $G^\omega$.
For the corresponding quadratic form,  $\mathcal{E}(\omega_1, \omega_2)=(dG^{\omega_1}, dG^{\omega_2})_1$, we get,
after a partial integration and on using (\ref{Hodgeomega}), (\ref{normalizationG}), 
\begin{equation}\label{Eomega}
\mathcal{E}(\omega_1, \omega_2)=\int_M dG^{\omega_1} \wedge *dG^{\omega_2}=\int_M G^{\omega_1} \wedge \omega_2.
\end{equation}


\subsection{The one point Green function}\label{sec:one point green}

The mutual energy and many other concepts in Hodge theory extend in a straight-forward manner to circumstances in which some source distributions have infinite energy.
This applies in particular to the Dirac current $\delta_a$, which we consider as a $2$-form with distributional coefficient, namely defined by the property that
$$
\int_M \delta_a \wedge \varphi =\varphi(a)
$$ 
for every smooth function ($0$-form) $\varphi$.
Certainly $\delta_a$ has infinite energy, but if $a\ne b$, then $\mathcal{E}(\delta_a, \delta_b)$ is still finite and has a natural interpretation:
it is the {\it Green function}: 
\begin{equation}\label{Gab}
G(a,b)=G^{\delta_a}(b)= \mathcal{E}(\delta_a, \delta_b).
\end{equation}
Here the first equality can be taken as a definition, and then the second equality
follows on using (\ref{Hodgeomega}), (\ref{normalizationG}):
\begin{align*}
  \mathcal{E}(\delta_a, \delta_b) &= \int_M dG^{\delta_a}\wedge *dG^{\delta_b}= -\int_M G^{\delta_a}\wedge d*dG^{\delta_b} \\
                                  &= \int_M G^{\delta_a}\wedge \left( \delta_b-\frac{1}{{\rm vol}^2(M)}\,{\rm vol}^2 \right) = G^{\delta_a}(b) = G(a,b).
\end{align*}
This shows in addition that $G(a,b)$ is symmetric.


\subsection{The two point Green function (``fundamental potential'')}

The background homogenous sink $c_\omega {\rm vol}^2$ in Section~\ref{sec:one point green} disappears in the case of several poles with net mass zero. The simplest case is the two
point fundamental potential, with one source and one sink, equally strong. This potential can be traced back to Riemann's original ideas
and it is explicitly used for example in \cite{Schiffer-Spencer-1954}.
In contrast to the one point Green function, the two point Green function does not depend on the Riemannian metric.

To conform with some previous usage we define the fundamental potential without the factor $1/2\pi$ in front of the logarithm.
We may define it in terms of the one point Green function as
\begin{equation}\label{VGGGG}
V(z,w;a,b) = 2\pi \left(G(z,a)-G(z,b)- G(w,a)+G(w,b)\right)=
\end{equation}
$$
=-\log |z-a|+\log |z-b| +\text{harmonic},
$$
where $a,b,w\in M$ are distinct points. It is mainly considered as a function of $z$, normalized so that it vanishes for $z=w$, and it is 
useful for discussing harmonic and holomorphic differentials in general. Indeed,
$$
 -dV(z)-\I*dV(z)=-2\,\frac{\partial V(z,w;a,b)}{\partial z}dz=
$$
$$
=\frac{dz}{z-a}-\frac{dz}{z-b} +\text{analytic}
$$
is the unique Abelian differential of the third kind with poles of residues $\pm 1$ at $z=a,b$ and with
purely imaginary periods. It does not depend on $w$, as can be seen from the decomposition (\ref{VGGGG}).

Obviously $dV$ is exact, while $*dV$ has certain periods around closed curves.
To be precise, if $\alpha$ is a closed curve then the function
\begin{equation}\label{Phioint}
U_\alpha (a,b)=\frac{1}{2\pi} \oint_\alpha *dV(\cdot,w;a,b)
\end{equation}
is, away from $\alpha$, harmonic in each of $a$ and $b$, and it makes a unit additive jump as one of these variables crosses $\alpha$.
The jump with respect to the variable $a$ is $-1$ when $a$ crosses $\alpha$ from the right to the left.
Assuming that $a$ moves along another closed curve $\beta$, which has no further intersections with $\alpha$, it follows that $U_\alpha(a,b)$,
as a compensation for the jump, has to  increase by $+1$ as $a$ runs through the rest of $\beta$. This is what lies behind the period relations 
(\ref{taualpha}), (\ref{taubeta}) below. 

Differentiating  $U_\alpha (a,b)$ with respect to $a$ and ignoring the distributional contribution from the jump
one obtains a differential 
$$
dU_\alpha(a)=\frac{\partial U_\alpha(a,b)}{\partial a}da+\frac{\partial U_\alpha(a,b)}{\partial \bar{a}}d\bar{a},
$$
which is harmonic (with respect to $a$) in all $M$. It is certainly harmonic in $b$ too, but it can be seen from (\ref{VGGGG}) that it actually does not depend on $b$.
For its integral we have
\begin{equation}\label{Uab}
U_\alpha(a,b)=\int_b^a dU_\alpha
\end{equation}
if we integrate along a path that does not intersect $\alpha$.

Let $\{\alpha_1, \dots, \alpha_\texttt{g}, \beta_1,\dots, \beta_\texttt{g}\}$ be representing cycles for a canonical homology basis for $M$ 
such that each $\beta_j$ intersects $\alpha_j$ once from the right to the left and no other intersections occur (see \cite{Farkas-Kra-1992} for details).
The harmonic differentials $dU_{\alpha_j}$, $dU_{\beta_j}$ obtained by the above construction then constitute (when taken in appropriate order
and with appropriate signs) the canonical basis of harmonic differentials associated with the chosen homology basis.
Precisely we have, for $k,j=1,\dots,\texttt{g}$,
\begin{equation}\label{taualpha}
\oint_{\alpha_k} (-dU_{\beta_j})=\delta_{kj}, \quad \oint_{\beta_k} (-dU_{\beta_j})=0,
\end{equation}
\begin{equation}\label{taubeta}
\oint_{\alpha_k} dU_{\alpha_j}=0, \qquad \oint_{\beta_k} dU_{\alpha_j}=\delta_{kj}.
\end{equation}

In terms of the above harmonic differentials the conjugate periods of $V$ can by (\ref{Phioint}), (\ref{Uab}) be explicitly expressed as
\begin{equation}\label{ointdVa}
\frac{1}{2\pi} \oint_{\alpha_j} *dV(\cdot, w;a,b)=\int_b^a dU_{\alpha_j},
\end{equation}
\begin{equation}\label{ointdVb}
\frac{1}{2\pi} \oint_{\beta_j} *dV(\cdot, w;a,b)=\int_b^a dU_{\beta_j}.
\end{equation}
Here the integration in the right member is to be performed along a path that does not intersect the curve in the left member.

One point with the above formulas is that they explicitly exhibit the dependence of the conjugate periods of  $V$
on the locations $a$, $b$ of the poles. This is essential because,
in our applications, $V$ will have the role of  being (part of) a stream function, and the conjugate periods of $V$ 
will enter into the circulations of the flow, which are to be  conserved in time by Kelvin's law (\ref{ddtgamma}).


\section{Affine and projective connections}\label{sec:connections}

Besides differential forms, and tensor fields in general,
affine and projective connections are quantities on Riemann surfaces which are relevant for point vortex motion, so we expand somewhat
on these notions below. The affine connections have the same meanings as in ordinary differential geometry, used to define 
covariant derivatives for example, and they play an important role in many areas of mathematical physics.

Some general references for the kind of connections we are going to consider are  
\cite{Schiffer-Hawley-1962, Gunning-1966, Gunning-1967, Gunning-1978, Gustafsson-Peetre-1989, Gustafsson-Sebbar-2012}.  The ideas go back at least to E.~Cartan.
We shall define them in the simplest possible manner, namely as quantities defined in terms of local holomorphic coordinates and transforming in
specified ways when changing from one coordinate to another.

Let $\tilde{z}=\varphi(z)$ represent a holomorphic local change of complex coordinate on $M$ and define three nonlinear differential expressions
$\{\cdot,\cdot\}_k$, $k=0,1,2$, by
\begin{align*}
\{\tilde{z},z\}_0 &= \log \varphi'(z)&&=-2 \log \frac{1}{\sqrt{\varphi'}}\\
\{\tilde{z},z\}_1 &= (\log \varphi'(z))' =\frac{\varphi''}{\varphi'}&&= -2 \sqrt{\varphi'}\,\,(\frac{1}{\sqrt{\varphi'}})'\\
\{\tilde{z},z\}_2 &= (\log \varphi'(z))''-\frac{1}{2}((\log \varphi'(z))')^2=\frac{\varphi'''}{\varphi'}-\frac{3}{2}(\frac{\varphi''}{\varphi'})^2 &&= -2 \sqrt{\varphi'}\,\,(\frac{1}{\sqrt{\varphi'}})''
\end{align*}
The last expression is the \emph{Schwarzian derivative} of $\varphi$. For
$\{\tilde{z},z\}_0$ there is an additive indetermincy of $2\pi \I$, so
actually only its real part, or exponential, is completely well-defined.

Alternative definitions: With $\tilde{z}=\varphi(z)$, $\tilde{a}=\varphi(a)$, we have
\begin{align*}
\{\tilde{a},a\}_0&=\lim_{z\to a} \log\frac{\tilde{z}-\tilde{a}}{z-a},\\
\{\tilde{a},a\}_1&=2\lim_{z\to a} \frac{\partial}{\partial z}\log\frac{\tilde{z}-\tilde{a}}{z-a},\\
\{\tilde{a},a\}_2&=6\lim_{z\to a} \frac{\partial^2}{\partial z\partial a}\log\frac{\tilde{z}-\tilde{a}}{z-a}.
\end{align*}

\begin{remark}
Both $z$ and $a$ (and $\tilde{z}$ and $\tilde{a}$) will be viewed as local variables, but $a$ will be slightly ``less variable'' than $z$ in the sense
that it will be used as a local base point for Taylor expansions.
\end{remark}

The following chain rules hold, if $z$ depends on $w$ via an intermediate variable $u$:

\begin{equation*}
\{z,w\}_k (dw)^k=\{z,u\}_k (du)^k+\{u,w\}_k (dw)^k \quad (k=0,1,2).
\end{equation*}
In particular,
\begin{equation*}
\{z,w\}_k (dw)^k=-\{w,z\}_k (dz)^k \quad (k=0,1,2).
\end{equation*}
It turns out that the three operators $\{\cdot,\cdot\}_k$, $k=0,1,2$, are unique in having properties as above,
i.e., one can not go on with anything similar for $k=3,4,\dots$. See \cite{Gunning-1966, Gunning-1967} for details. 

\begin{definition}\label{def:affine}
An {\it affine connection} on $M$ is an object which is represented
by local differentials $r(z)dz$,
$\tilde{r}(\tilde{z})d\tilde{z}$,\dots (one in each coordinate
variable, and not necessarily holomorphic) glued together according to the rule
\begin{equation*}
\tilde{r}({\tilde{z}}){d\tilde{z}}=r(z){dz}-\{\tilde{z},z\}_1\,{dz}.
\end{equation*}
\end{definition}

\begin{definition}\label{def:projective}
A {\it projective connection} on $M$ consists of local
quadratic differentials $q(z)(dz)^2$,
$\tilde{q}(\tilde{z})(d\tilde{z})^2$, \dots, glued together
according to
\begin{equation*}
\tilde{q}({\tilde{z}})({d\tilde{z}})^2=q(z)({dz})^2-\{\tilde{z},z\}_2\,({dz})^2.
\end{equation*}
\end{definition}

One may also consider $0$-{\it connections},
quantities defined up to multiples of $2\pi \I$ and which transform
according to
\begin{equation*}
\tilde p (\tilde z)= p(z) -\{\tilde z, z\}_0.
\end{equation*}
This means exactly that $e^{p(z)}$ is well-defined and transforms as
differential of order one, and its absolute value $\lambda(z) =e^{\,{\rm Re\,} p(z)}$ transforms as
the coefficient a {Riemannian metric} when this is written as
\begin{equation*}
ds =\lambda(z){|dz|}.
\end{equation*}

If the Riemann surface initially is provided with a Riemannian metric $ds =\lambda(z){|dz|}$, then
there is a natural affine connection associated to it  by
\begin{equation}\label{rlog}
r(z)= 2 \frac{\partial }{\partial z}\log \lambda(z)=\frac{\partial \log\lambda}{\partial x}-\I \frac{\partial \log \lambda}{\partial y}.
\end{equation}
This is the same as the Levi-Civita connection in general tensor analysis,
and the real and imaginary parts, made explicit above, coincide (up to sign) with the components of the classical Christoffel symbols $\Gamma_{ij}^k$.
Compare formulas in  \cite{Grotta-Ragazzo-Barros-Viglioni-2017}. 

Independent of any metric, an affine connection $r$ gives rise to a projective connection $q$ by
\begin{equation}\label{qr}
q= \frac{\partial r}{\partial z}- \frac{1}{2}r^2.
\end{equation}
This $q$ is sometimes called the ``curvature" of $r$ (see \cite{Dubrovin-1993}). That curvature is however not the same as the Gaussian curvature in case
$r(z)$ comes form a metric. 

In the presence of an affine connection one can define,
for every half integer $k\in\frac{1}{2}{\mathbb{Z}}$, a {covariant derivative}
$\nabla_k$ taking $k$:th order differentials to $(k+1)$:th order
differentials by $\phi(dz)^k\mapsto (\nabla_k \phi)(dz)^{k+1}$, where
\begin{equation*}\label{eq:affderivative}
\nabla_k \phi = \frac{\partial\phi}{\partial z}-kr\phi.
\end{equation*}
The covariance means that
$$
\phi(dz)^k=\tilde{\phi}(d\tilde{z})^k \Longrightarrow
(\nabla_k\phi)(dz)^{k+1}=(\tilde{\nabla}_k\tilde{\phi})(d\tilde{z})^{k+1}
$$

A projective connection also allows for certain covariant derivatives: for each $m=0, 1, 2,
\dots$ there is a linear differential operator $\Lambda_m$ taking
differentials of order $\frac{1-m}{2}$ to differentials of order
$\frac{1+m}{2}$:
$$
\Lambda_m: \phi (dz)^{\frac{1-m}{2}}\mapsto (\Lambda_m\phi)
(dz)^{\frac{1+m}{2}}.
$$
In case the projective connection comes from an affine connection as in (\ref{qr}) these $\Lambda_m$ are given by 
\begin{equation*}\label{eq:Lnablanabla}
\Lambda_1=\nabla_0, \quad \Lambda_2
=\nabla_{\frac{1}{2}}\nabla_{-\frac{1}{2}}, \quad
\Lambda_3=\nabla_{1}\nabla_{0}\nabla_{-1},\quad {\rm etc}.
\end{equation*}

\begin{example}
For $m=2$ we have, with $q$ given by (\ref{qr}),
$$
\Lambda_2 (\phi) =(\frac{\partial}{\partial z}-\frac{r}{2})(\frac{\partial}{\partial z}+\frac{r}{2})\phi=\frac{\partial^2\phi}{\partial z^2}+\frac{1}{2}q\phi.
$$
The right member here is what $\Lambda_2 (\phi)$ looks like even if $q$ does not come from any affine connection.
We refer to \cite{Gustafsson-Peetre-1989, Gustafsson-Sebbar-2012} for further details and more examples.
\end{example}


\section{Behavior of singular parts under changes of coordinates}\label{sec:singular parts}

The different kinds of connections discussed above appear naturally when studying
expansions of singular parts of Green functions and stream functions under changes of coordinates.

To give an intuitive motivation, consider a point vortex with location $a$, moving in a planar domain,
possibly together with other vortices. The stream function $\psi$ of the flow (see Section~\ref{sec:point vortices}),
will then have a local expansion starting (up to a constant factor depending on the strength of the vortex)
$$
\psi(z)=\log|z-a|-c_0(a) -\mathcal{O}(|z-a|).
$$
Similarly, the expansion of the analytic completion $\nu+\I*\nu=2\I\frac{\partial\psi}{\partial{z}}dz$ of the flow $1$-form $\nu$ will look like (again up to a constant factor)
$$
f(z)dz=\frac{dz}{z-a} -c_1(a)dz -\mathcal{O}(z-a).
$$ 

One step further, one may consider vortex dipoles, i.e., two vortices of opposite strengths infinitesimally close to each other. The corresponding flow is
obtained by taking the derivative $\partial/\partial a$ of the above flow $1$-form and then multiply
this by a complex factor to adjust for the orientation of the dipole. Taking this factor to be one, for simplicity, we obtain something of the form
$$
F(z)dzda= \frac{dzda}{(z-a)^2}-2c_2(a)dzda + \mathcal{O}(z-a),
$$
and by the same argument as above the motion of the dipole should then be determined by the coefficient $2c_2(a)$. Lemma~\ref{lem:connections}
below says that this coefficient is (essentially) a projective connection. However, vortex dipoles are very singular, and have to be renormalized in a special way
in order to prevent them having infinite speed. Compare \cite{Llewellyn-Nagem-2013}.

Now the general statement in this context is the following.

\begin{lemma}\label{lem:connections}
Let $\tilde{z}=\varphi(z)$ be a local conformal map representing a change of coordinates near a point $z=a$, and set $\tilde{a}=\varphi (a)$.
Let $\psi(z)$ ($=\psi(z,a)$) be a locally defined harmonic function with a logarithmic pole at $z=a$,  similarly $f(z)dz$ a meromorphic differential with a simple pole 
at the same point, and $F(z)dzda$ a double differential with a pure (residue free) second order pole.
Precisely, we assume the following local forms, in the $z$ and $\tilde{z}$ variables:
\begin{align*}
\psi(z)&= \log|z-a|-c_0(a) + \mathcal{O}(|{z}-{a}|)\\
&=\log|\tilde{z}-\tilde{a}|-\tilde{c}_0(\tilde{a})+ \mathcal{O}(|\tilde{z}-\tilde{a}|),\\
f(z)dz&= \frac{dz}{z-a}-c_1(a)dz + \mathcal{O}(z-a)\\
&=\frac{d\tilde{z}}{\tilde{z}-\tilde{a}}-\tilde{c}_1(\tilde{a})d\tilde{z}+ \mathcal{O}(\tilde{z}-\tilde{a}),\\
F(z)dzda&= \frac{dzda}{(z-a)^2}-2c_2(a)dzda + \mathcal{O}(z-a)\\
&=\frac{d\tilde{z}d\tilde{a}}{(\tilde{z}-\tilde{a})^2}-2\tilde{c}_2(\tilde{a})d\tilde{z}d\tilde{a}+ \mathcal{O}(\tilde{z}-\tilde{a}).
\end{align*}
Then, as functions of the location of the singularity and up to constant factors, $c_0$ transforms as the real part of a $0$-connection,
$c_1$ as an affine connection and $c_2$ as a projective connection:
$$
\tilde{c}_0(\tilde{a})= c_0(a)+\re\{\tilde{a},a\}_0, 
$$
$$
\tilde{c}_1(\tilde{a}) d\tilde{a}= c_1 (a)da +\frac{1}{2} \{\tilde{a}, a\}_1 da,
$$
$$
2\tilde{c}_2(\tilde{a}) d\tilde{a}^2= 2c_2 (a)da^2 +\frac{1}{6} \{\tilde{a}, a\}_2 da^2.
$$
The first statement is most conveniently expressed by saying that
$$
ds=  e^{-\tilde{c}_0(\tilde{a})}|d\tilde{a}|=e^{-c_0(a)}|da|
$$
defines a conformally invariant metric.
\end{lemma}

\begin{remark}
It is actually not necessary for the conclusions of the lemma that $\psi$, $fdz$, $Fdzda$ are harmonic/analytic away from the singularity,
it is enough that the local forms of the singularity and constant terms given as above hold.
\end{remark}

\begin{proof}[Proof {\rm (of Lemma)}]
The proof consists of straight-forward substitutions, but let us for clarity give a few details. 
The needed ingredients are the equation $d\tilde{z}=\varphi'(z)dz$ and the Taylor series of $\varphi$,
which we write as
$$
\frac{\tilde{z}-\tilde{a}}{z-a}=\varphi'(a)+ \frac{1}{2}\,\varphi''(a)(z-a)+\frac{1}{6}\varphi'''(a)(z-a)^2+\dots.
$$
The statement for $\psi$ is immediate and,  for $fdz$  one obtains, after multiplying
the two given expressions for $fdz$ by $\tilde{z}-\tilde{a}$ and dividing by $dz$, that  
$$
\varphi'(a)+\frac{1}{2}\,\varphi''(a)(z-a) -c_1(a)(\tilde{z}-\tilde{a})=
$$
$$
= \varphi'(z)-\varphi'(z)\tilde{c}_1(\tilde{a})(\tilde{z}-\tilde{a})+\mathcal{O}((z-a)^2)
$$
as $z\to a$. After dividing also with $z-a$ and letting $z\to a$ one obtains
$$
\tilde{c}_1(\tilde{a})\varphi'(a)=c_1(a)+\frac{\varphi''(a)}{2\varphi'(a)},
$$
which is the desired result.

As for $Fdzda$ one uses also that $d\tilde{z}d\tilde{a}=\varphi'(z)\varphi' (a)dzda$ and that
$$
(\frac{\tilde{z}-\tilde{a}}{z-a})^2=\varphi'(a)^2+ \varphi'(a)\varphi''(a)(z-a)+
$$
$$
+\left(\frac{1}{4}\varphi''(a)^2+\frac{1}{3}\varphi'(a)\varphi'''(a)\right)(z-a)^2+\mathcal{O}((z-a)^3).
$$
Inserting this into the two given expressions for $Fdzda$ and comparing the Taylor expansions,
after multiplication with $(\tilde{z}-\tilde{a})^2$ and division with $dzda$, gives that the constant
and linear  terms match automatically, while matching of the $(z-a)^2$ terms requires that
$$
2\tilde{c}_2(\tilde{a})\varphi'(a)^2=2c_2(a)+\frac{1}{6}\left( \frac{\varphi'''(a)}{\varphi'(a)}  -\frac{3}{2}\frac{\varphi''(a)^2}{\varphi'(a)^2}\right).
$$
And this is exactly the desired identity.

\end{proof}

\begin{remark}\label{rem:nonpure}
For a ``nonpure'' second order pole, with expansion
\begin{align*}
Fdzda&= \frac{dzda}{(z-a)^2} +\frac{b_2(a)dzda}{z-a}-2c_2(a)dzda + \mathcal{O}(z-a)\\
&=\frac{d\tilde{z}d\tilde{a}}{(\tilde{z}-\tilde{a})^2} +\frac{\tilde{b}_2(\tilde{a})d\tilde{z}d\tilde{a}}{\tilde{z}-\tilde{a}}-2\tilde{c}_2(\tilde{a})d\tilde{z}d\tilde{a}+ \mathcal{O}(\tilde{z}-\tilde{a}), 
\end{align*}
the new coefficient $b_2$ transforms as an ordinary differential:
$$
\tilde{b}_2(\tilde{a})d\tilde{a}=b_2(a)da,
$$
while $c_2(a)$ continues to transform as an projective connection (up to a constant). The proof is again straight-forward. 
\end{remark}

We may adapt the above to the Green function $G(z,a)=G^{\delta_a} (z)$, despite it  is not harmonic in $z$
(because of the compensating background flow).  First we write
\begin{equation}\label{GlogH}
G(z,a)=\frac{1}{2\pi} (-\log |z-a| +H(z,a))
\end{equation}
and then expand the regular part as
\begin{equation}\label{greentaylor}
H(z,a)= h_0 (a)+\frac{1}{2}\left(h_1(a)(z-a)+ \overline{h_1 (a)} (\bar{z}-\bar{a})\right)+
\end{equation}
$$
+\frac{1}{2}\left(h_{2}(a)(z-a)^2+\overline{h_{2}(a)}(\bar{z}-\bar{a})^2\right)+h_{11}(a)(z-a)(\bar{z}-\bar{a})+ \mathcal{O}(|z-a|^3).
$$

The coefficient $h_0(a)$ is one example of a {\it Robin function}, more precisely a {\it coordinate Robin function}.
The distance $|z-a|$ to the singularity is here measured using a coordinate chart, this explaining why $h_0(a)$
is not a true function, only a kind of connection. More common is to use the invariant (geodesic) distance 
in the logarithm, and then the Robin function really becomes a function.   

\begin{lemma}\label{lem:taylor}
Some relations between the coefficients in (\ref{greentaylor}) are
\begin{equation}\label{c1c0}
h_1(a)=\frac{\partial h_0 (a)}{\partial a},
\end{equation}
\begin{equation}\label{c2c1c0}
h_2(a)=\frac{1}{2}\frac{\partial h_1(a)}{\partial a}-\{\frac{\partial^2 H(z,a)}{\partial z \partial a}\}_{z=a}.
\end{equation}\end{lemma}
\begin{equation}\label{c11lambda}
h_{11}(a)=\frac{1}{2}\frac{\partial h_1(a)}{\partial \bar{a}}-\{\frac{\partial^2 H(z,a)}{\partial z \partial \bar{a}}\}_{z=a}= \frac{\pi\lambda(a)^2}{2 {\rm vol}^2 (M)},
\end{equation}

\begin{proof}
The proof depends on $G(z,a)$, and hence $H(z,a)$, being symmetric in the two variables and that, as an instance of (\ref{Hodgeomega}),
\begin{equation}\label{ddGdelta}
-d*d G(\cdot, a)=\delta_a-\frac{1}{{\rm vol}^2(M)}\, {\rm vol}^2.
\end{equation}
Indeed, (\ref{c1c0}) is a consequence of
$$
h_1(a)= 2\frac{\partial}{\partial z}\big|_{z=a} H(z,a) =  \frac{\partial }{\partial a} H(a,a),
$$
and (\ref{c2c1c0}) is obtained by differentiating (\ref{greentaylor}) with respect to $z$ and $a$ and letting $z\to a$.
The final term in (\ref{c11lambda}) comes out on comparing
$$
\{\frac{\partial^2 H(z,a)}{\partial z \partial\bar{z}}\}_{z=a} = h_{11}(a)
$$
with the equation for $H(z,a)$ which is embodied in (\ref{ddGdelta}).

\end{proof}

\begin{example}\label{ex:sphere}
In the special case of  the Riemann sphere, identification with the extended plane
$M= \C\cup \{\infty\}$ under stereographic projection renders the spherical metric on the form
$$
ds^2=\frac{4|dz|^2}{(1+|z|^2)^2}.
$$
This means that $\lambda(z)^2={4}/{(1+|z|^2)^2}$, ${\rm vol^2}(M)=4\pi$, and we have
$$
G(z,a)=-\frac{1}{4\pi}\left(   \log \frac{|z-a|^2}{(1+|z|^2)(1+|a|^2)} +1\right).
$$
Hence
$$
H(z,a)=\frac{1}{2}\log (1+|z|^2)+\frac{1}{2}\log (1+|a|^2)-\frac{1}{2},
$$
so that, referring to (\ref{greentaylor}),
$$
h_0(a)=H(a,a)= \log (1+|a|^2) -\frac{1}{2},
$$
$$
h_1(a)=2\frac{\partial }{\partial z}|_{z=a} H(z,a)=\frac{\bar{a}}{1+|a|^2},  
$$
$$
h_2(a)=-\frac{\bar{a}^2}{2(1+|a|^2)^2}.
$$
$$
h_{11}(a)=\frac{\partial^2 }{\partial z\partial \bar{z}}|_{z=a} H(z,a)=\frac{{1}}{2(1+|a|^2)^2}.
$$
The metric associated to the Green function becomes
$$
ds^2 = e^{-2h_0(a) |da|^2}=\frac{{e}\,|da|^2}{(1+|a|^2)^2},
$$
hence equals the spherical metric, up to a factor.

The fundamental potential (\ref{VGGGG}) on the Riemann sphere is, independently of the metric,
$$
V(z,w;a,b)=-\log\big|\frac{(z-a)(w-b)}{(z-b)(w-a)}\big|,
$$
i.e., (minus) the logarithm of the modulus of the cross ratio. 

\end{example}

The transformation properties of the coefficients $h_j(a)$ are closely related
to those appearing in Lemma~\ref{lem:connections}: 

\begin{lemma}\label{lem:gammaconnections}
Under a local holomorphic change 
$\tilde{z}=\varphi(z)$ of coordinates, with $\tilde{a}=\varphi (a)$, we have
\begin{equation}\label{h0}
\tilde{h}_0(\tilde{a})= h_0(a)+\re\{\tilde{a},a\}_0, 
\end{equation}
\begin{equation}\label{h1}
\tilde{h}_1(\tilde{a}) d\tilde{a}= h_1 (a)da +\frac{1}{2} \{\tilde{a}, a\}_1 da,
\end{equation}
\begin{equation}\label{h2h1}
\big(\frac{\partial \tilde{h}_1(\tilde{a})}{\partial \tilde{a}}-2\tilde{h}_2(\tilde{a}) \big)d\tilde{a}^2
= \big(\frac{\partial h_1(a)}{\partial a}-2{h}_2(a) \big)d{a}^2 +\frac{1}{6} \{\tilde{a}, a\}_2 da^2,
\end{equation}
\begin{equation}\label{h11}
\tilde{h}_{11}(\tilde{a})|d\tilde{a}|^2=h_{11}(a)|da|^2.
\end{equation}
\end{lemma}

\begin{remark}\label{rem:h}
We see that both $h_0$ and $h_{11}$ define Riemannian metrics on $M$. For $h_{11}$ this is explicit in (\ref{h11}),
and for $h_0$ the metric is
\begin{equation}\label{h0metric}
ds^2= e^{-2h_0(a)}|da|^2.
\end{equation}
Compare the last equation in Lemma~\ref{lem:connections}. However, the two metrics are in general not the same: as is seen from
(\ref{c11lambda}) the metric $ds^2=h_{11}(a)|da|^2$ equals, up to a constant factor, the initially given metric on $M$,
while (\ref{h0metric}) usually is a different metric, somewhat related to (but not identical with)
what in \cite{Grotta-Ragazzo-Barros-Viglioni-2017} is called the ``steady vortex metric''. 

However it should be noted that the metric (\ref{h0metric}) actually depends on the initially given metric. But if the two metrics agree up to a constant factor,
as in the case of the sphere above (in Example~\ref{ex:sphere}),
then the situation is stable and, as will be a consequence of Theorem~\ref{thm:dzkdt} below, these metrics then will have the property that a single vortex does not
move unless one adds circulations beyond those which are already present in the Green function itself.

For later reference we point out that, by (\ref{rlog}), $-\frac{\partial}{\partial a}\log \lambda(a)da$ transforms in the same way as $h_1(a)da$.
 
\end{remark}

\begin{proof}[Proof {\rm (of Lemma)}]
The proofs of the first two equations are similar to those of the corresponding statements in Lemma~\ref{lem:connections}.
The third statement comes out on combining Lemma~\ref{lem:connections} with (\ref{c2c1c0}) in Lemma~\ref{lem:taylor},
noticing that
$$
-4\pi \frac{\partial^2 G(z,a)}{\partial z \partial a}dzda = \frac{dzda}{(z-a)^2} -2\frac{\partial^2 H(z,a)}{\partial z \partial a}dzda,
$$ 
where the left member behaves as the coefficient of a double differential of type $dzda$.
It follows that $2\{\frac{\partial^2 H(z,a)}{\partial z \partial a}\}_{z=a}$ transforms as $2c_2(a)$ in Lemma~\ref{lem:connections}.

The last statement is immediate from (\ref{c11lambda}).
\end{proof}


\section{A weak formulation of Euler's equation}\label{sec:euler}

We return to fluid mechanics and specialize to incompressible fluid flow in two dimensions. 
Thus $M$ is a Riemannian manifold of dimension $n=2$, and  we take the constitutive equation 
(\ref{constitutive}) to be simply that $\rho=1$.
Then the equation of continuity (\ref{continuityeq}) becomes very simple, namely
\begin{equation}\label{dstarnu}
d*\nu=0.
\end{equation}
Indeed, the fluid velocity vector ${\bf v}$ and the corresponding covector ($1$-form) $\nu$ are related by (\ref{Hodgei}),
and so $d*\nu =d(i_{\bf v}{\rm vol}^2)=\mathcal{L}_{\bf{v}}({\rm vol}^2) =0$  by (\ref{cartan}), (\ref{continuityeq}).

It follows that $*\nu$ locally is of the form 
$$
*\nu=d\psi
$$
for some {\it stream function} $\psi$. However, this $\psi$ will not be globally single-valued in general. For any branch of $\psi$ we have 
$-d*d\psi =\omega$ in terms of the vorticity $2$-form $\omega=d\nu$. Clearly  $\int_M \omega=\int_M d\nu=0$, so (\ref{Hodgeomega}), (\ref{tvol}) give
$$
\omega= -d*d G^\omega, 
$$
where $G^\omega$ has the advantage of being single-valued. It follows that 
\begin{equation}\label{nudG}
\nu = -*dG^\omega + \eta,
\end{equation}
for some $1$-form  $\eta$ satisfying $d\eta =0$. Using  (\ref{dstarnu}) it follows that also $d*\eta=0$, i.e.,
$\eta$ is a harmonic $1$-form. As a such, it is uniquely determined by its periods
$\oint_{\alpha_k} \nu$, $\oint_{\beta_k} \nu$ with respect to a canonical homology bases $\{\alpha_k, \beta_k: k=1,\dots, {\texttt{g}}\}$.

Setting 
$$
a_k=\oint_{\alpha_k}\nu, \quad b_k=\oint_{\beta_k}\nu
$$
for the periods it follows that there is a one-to-one correspondence
$$
\nu \leftrightarrow (\omega, \{a_k, b_k\})
$$
between the velocity $1$-form, on one hand, and the vorticity $2$-form combined with given periods on the other hand.

The  following three steps give an equation for $\frac{\partial\omega}{\partial t}$ in terms of $\omega$:
\begin{itemize}
\item $(\omega,  \{a_k, b_k\})\mapsto \nu$ by Hodge theory, i.e., by solving
$$
\begin{cases}
d\nu =\omega,\\
d*\nu =0,\\
\oint_{\alpha_k} \nu=a_k, \quad \oint_{\beta_k} \nu=b_k.
\end{cases}
$$
\item $\nu \mapsto {\bf v}$ by using the metric:
if $\nu =v_x dx+ v_y dy$, then $*\nu= -v_y dx+v_x dy$ and
$$
{\bf v}= \frac{1}{\lambda^2} (v_x \frac{\partial}{\partial x}+ v_y \frac{\partial}{\partial y}).
$$
\item ${\bf v}\mapsto \frac{\partial \omega}{\partial t}$
by Euler's equations, or conservation of vorticity:
\begin{equation}\label{omegaLv}
\frac{\partial\omega}{\partial t}+\mathcal{L}_{\bf{v}}(\omega)=0.
\end{equation}
\end{itemize}

As for the first item, it may be noticed that having the data $(\omega,  \{a_k, b_k\})$ is equivalent to knowing
the periods $\oint_\gamma \nu$ for all closed curves $\gamma$ in $M$. Thus the first step can be viewed as a map
$$
\{\oint_\gamma \nu: \text{all closed curves }\gamma\}\mapsto \nu.
$$
Expressing (\ref{omegaLv}) directly in terms of $\nu$ (i.e., avoiding the step via ${\bf v}$) it becomes
$$
\frac{\partial{\omega}}{\partial t}  =*\nu\wedge d *{\omega}.
$$
Here the right member is an antisymmetric bilinear form in $\nu$ and $\omega$, a form which is amply discussed from 
Lie algebra points of view in \cite{Arnold-Khesin-1998}.

We can complete $\nu$ into a complex-valued $1$-form by
$$
\nu + \I *\nu=-*d\psi +\I d\psi=\I (dG^\omega+ \I *dG^\omega)+\eta + \I *\eta.
$$
In regions where $\omega=0$ we can also write $\nu =d\varphi$, where $\varphi$ is a locally defined velocity potential. Then
$\nu + \I *\nu=d\varphi +\I d\psi=d\Phi$, with 
\begin{equation}\label{Phi} 
\Phi =\varphi + \I \psi.
\end{equation}
We emphasize that $\Phi$ makes sense only in regions with no vorticity, and that it then is an additively multi-valued analytic function.
In the point vortex case $d\Phi$ is however a well-defined meromorphic differential (Abelian differential). 

It is straight-forward to formulate Euler's equation (\ref{euler}) directly in terms of $\psi$ and $p$.
Using dot for time derivative the result is the somewhat awkward looking equation
\begin{equation}\label{euler0}
-*d\dot{\psi} + \frac{1}{2}d(*(d\psi\wedge *d\psi)) +dp= (*d*d\psi) \wedge d\psi.
\end{equation}
This is an equation for real-valued $1$-forms, which can be expressed in terms of a local basis $dx$, $dy$ and then becomes one equation for each 
component.  Nothing prevents us from
using complex-valued $1$-forms, with natural basis $dz=dx+\I dy$, $d\bar{z}=dx+\I dy$. Then real-valued forms may get complex-valued coefficients.
For example $d\psi =\frac{\partial \psi}{\partial x}dx+\frac{\partial \psi}{\partial y}dy =\frac{\partial \psi}{\partial z}dz+\frac{\partial \psi}{\partial \bar{z}}d\bar{z}$,
with then $\frac{\partial \psi}{\partial z}$ and $\frac{\partial \psi}{\partial \bar{z}}$ complex-valued, but complex conjugates of each other.
The Hodge star acting on the above complex basis just multiplies with $\pm\I$: $*dz=-\I dz$, $*d\bar{z}=\I d\bar{z}$.

If we write (\ref{euler0}) in terms of the basis $\{dz,d\bar{z}\}$ then the coefficient of  $dz$ contains the same information as that of $d{\bar{z}}$, hence it is
enough to use one of them. Selecting to work with the $d\bar{z}$-component of (\ref{euler0}) then gives the equation
\begin{equation}\label{eulercomplex}
\frac{\partial}{\partial \bar z}\Big( \frac{2}{\lambda^2}\frac{\partial \psi}{\partial z}\frac{\partial\psi }{\partial \bar z}
+{p}-\I\dot{\psi} \Big)=\frac{2}{\lambda^2}\frac{\partial}{\partial z}(\frac{\partial \psi}{\partial \bar z})^2
\end{equation}
Here the right member can also be expressed as
$$
\frac{2}{\lambda^2}\frac{\partial}{\partial z}(\frac{\partial \psi}{\partial \bar z})^2
=\frac{4}{\lambda^2}\frac{\partial^2\psi}{\partial z\partial\bar{z}}\frac{\partial \psi}{\partial \bar z}
=-(*\omega)\,\frac{\partial \psi}{\partial \bar z}.
$$

We summarize:

\begin{lemma}\label{lem:f}
In terms of the (in general multi-valued) function
\begin{equation}\label{f}
f(z,t)=\frac{2}{\lambda^2}\frac{\partial \psi}{\partial z}\frac{\partial\psi }{\partial \bar z} +p- \I \frac{\partial \psi}{\partial t},
\end{equation}
the Euler equation (\ref{euler}) reads
\begin{equation}\label{euler2}
\frac{\partial f}{\partial \bar z}=\frac{2}{\lambda^2}\frac{\partial}{\partial z}(\frac{\partial \psi}{\partial \bar z})^2.
\end{equation}
\end{lemma}

We see from (\ref{euler2}) that $f$ is analytic in regions without vorticity, and in such regions it simply equals $\dot{\Phi}$
(up to sign):
$$
f=-\frac{\partial \Phi}{\partial t} \quad \text{in } \{\omega=0\}.
$$
We also see that 
\begin{equation}\label{bernuolli}
-\frac{\partial \varphi}{\partial t}=\frac{2}{\lambda^2}\frac{\partial \psi}{\partial z}\frac{\partial\psi }{\partial \bar z}+{p }=\frac{1}{2}|{\bf v}|^2+p
\end{equation}
in such regions. In other words, the time derivative of the velocity potential is a single-valued function having a direct hydrodynamical 
interpretation.   The fact that $\dot{\varphi}$ is globally single-valued can be seen as
a way to express Kelvin's law (\ref{ddtgamma})  for conservation of circulations. In contrast,
the imaginary part of $f$  need not be globally single-valued (but the derivative $\partial f/\partial \bar{z}$ always is).

As a further remark, Bernoulli's equation, stating that the right member of (\ref{bernuolli}) is constant in the case of stationary flow,
comes out very nicely: when the flow is stationary, so that
$\dot{\psi}=0$, then $f$ is analytic and real-valued in regions without vorticity, hence equals a (real) constant there. 

In the point vortex case one may view the usage of $f$ as one way (out of many) to get rid of the pressure $p$. Indeed, the pressure
shows up only in the real part of $f$, and one can see its role there as just a way to make this real part equal to the harmonic conjugate
of (minus) the imaginary part, namely $\dot{\psi}$.

In the smooth case (\ref{euler2}) holds if and only if the integral over arbitrary subdomains $D\subset M$ of the two members agree:
\begin{equation}\label{euler5}
\int_D d( fdz)  = \int_D \frac{2}{\lambda^2}\frac{\partial}{\partial z}(\frac{\partial \psi}{\partial \bar z})^2d\bar{z}dz.
\end{equation}
Strictly speaking, $D$ should be sufficiently small, since $f$ may be multivalued, and have smooth boundary to allow for partial integration in a next step.
In case $\lambda$ is constant (\ref{euler5}) leads immediately to
\begin{equation*}\label{euler1}
\int_{\partial D}  fdz  =- \int_{\partial D} \frac{2}{\lambda^2}\overline{(\frac{\partial \psi}{\partial  z})^2d{z}},
\end{equation*}
and in the point vortex limit both members can be evaluated by residues to give the well-known equations of motion for the vortices.

When $\lambda$ is not constant the same idea still works, but the details become slightly more involved. We assume that $\lambda$ is sufficiently smooth,
say real analytic for simplicity. Selecting a point $a\in D$, which when passing to the point vortex limit (in the next section) will be a vortex point, we expand
$$
\frac{2}{\lambda(z)^2}=\frac{2}{\lambda(a)^2}-\frac{4}{\lambda(a)^3}\frac{\partial\lambda(a)}{\partial a}\cdot(z-a)
-\frac{4}{\lambda(a)^3}\frac{\partial\lambda(a)}{\partial \bar{a}}\cdot(\bar{z}-\bar{a})+\mathcal{O}(|z-a|^2)
$$
$$
=\frac{2}{\lambda(a)^2}\Big(1-{2}\frac{\partial\log\lambda(a)}{\partial \bar{a}}\cdot(\bar{z}-\bar{a})\Big)+Q(z),
$$
where $Q(z)$ simply stands for the remainder, namely
$$
Q(z)=-\frac{4}{\lambda(a)^3}\frac{\partial\lambda(a)}{\partial a}\cdot (z-a)+\mathcal{O}(|z-a|^2).
$$

In (\ref{euler5}) one sees that all terms except the one with $Q(z)$ can be pushed to the boundary, i.e. (\ref{euler5}) can be written
$$
\int_{\partial D}  fdz  +  \frac{2}{\lambda^2(a)}\int_{\partial D}\overline{\big(1-{2}\frac{\partial\log\lambda(a)}{\partial {a}}\,({z}-{a})\big)(\frac{\partial \psi}{\partial  z})^2d{z}}=
$$
\begin{equation}\label{decomposition}
=\int_D Q(z)\frac{\partial}{\partial z}(\frac{\partial \psi}{\partial \bar z})^2d\bar{z}dz.
\end{equation}
When we pass to the limit that all vorticity inside $D$ is concentrated to the point $z=a$ the left member in (\ref{decomposition})
is easily evaluated by residues, only the right member has to be handled with some care.
This right member vanishes outside any (isolated) point vortex $z=a$, but there could be some kind of contribution exactly at $z=a$. One way to express the problem is 
to say that $({\partial \psi}/{\partial \bar z})^2$ is not locally absolutely integrable, hence does not make immediate sense as a distribution, and so it becomes 
difficult to make sense to the derivative $\partial/\partial z$ in front of it. 

The idea with the decomposition (\ref{decomposition}) is then that all contributions at $z=a$ have been sorted out to lie in the left member (the boundary integrals), and the right member
will not contribute at vortex points. And this is due to cancellations in the integrand (of the right member) which can be seen on writing it roughly as (one may think of smearing out the point
mass of $\Delta\psi$ a little)
$$
Q(z)\frac{\partial}{\partial z}(\frac{\partial \psi}{\partial \bar z})^2=2Q(z)\frac{\partial^2 \psi}{\partial z\partial\bar{z}}\frac{\partial \psi}{\partial \bar z}\sim
$$
$$
\sim(z-a)\cdot\Delta \psi\cdot \frac{1}{\bar{z}-\bar{a}}\sim ({\rm finite\ point \ mass})\cdot\frac{(z-a)^2}{|z-a|^2}.
$$
Due to cancellations when integrating the last factor one can safely write the right member of (\ref{decomposition}) as
$$
\lim_{\varepsilon\to 0}\int_{D\setminus\D(a,\varepsilon)} Q(z)\frac{\partial}{\partial z}(\frac{\partial \psi}{\partial \bar z})^2d\bar{z}dz,
$$
which then is stable under the limit process and makes good sense (and of course reduces to the full integral over $D$ in the smooth case).

The above will be good enough for our purposes, and it will lead to right point vortex dynamics, as given in \cite{Dritschel-Boatto-2015} for example, and
also to be confirmed by a Hamiltonian approach in Section~\ref{sec:hamiltoneq}. We summarize the weak formulation to be used as follows.

\begin{lemma}\label{lem:eulerweak}
When all data are smooth, the Euler equation (\ref{euler}) with $\rho=1$ is equivalent to the identity
\begin{equation}\label{eulerweak1}
\int_{\partial D}  fdz  +  \frac{2}{\lambda^2(a)}\int_{\partial D}\overline{\big(1-{2}\frac{\partial\log\lambda(a)}{\partial {a}}\,({z}-{a})\big)(\frac{\partial \psi}{\partial  z})^2\,d{z}}=
\end{equation}
$$
=\lim_{\varepsilon\to 0}\int_{D\setminus\D(a,\varepsilon)} Q(z)\frac{\partial}{\partial z}(\frac{\partial \psi}{\partial \bar z})^2d\bar{z}dz
$$
holding for all constellations  $a\in D\subset M$ with $D$ sufficiently small.
\end{lemma}

\begin{remark} 
There are many weak formulations of Euler's equation, giving different results,
or no result at all,  when passing to point vortex limit. So this passage is indeed nontrivial.
For example, it turns out that (\ref{eulerweak1}) is strong enough to give the speed of individual point vortices, but too weak to determine the motion 
of vortex dipoles. 
\end{remark}


\section{Motion of point vortices in terms of affine connections}\label{sec:point vortices}

In this section we shall combine the analysis of affine connections in Section~\ref{sec:connections} and the weak formulation of Euler's equation
in Section~\ref{sec:euler} to determine the equations of motion  for point vortices on a closed Riemannian manifold. The result, Theorem~\ref{thm:dzkdt},
is stated in a rather vague form, involving quantities which are not fully explicit, but this will be remedied in Sections~\ref{sec:hamiltonian} and
\ref{sec:hamiltoneq}, where the results are as explicit as they can possibly be in the general framework we are working with.
We start by summarizing the given data and assumptions in the point vortex case.

\begin{itemize}

\item $M$ is a compact Riemann surface of genus ${\texttt{g}}\geq 0$.

\item The vorticity is concentrated to finitely many points $z_1,\dots,z_n\in M$, so that: 
\begin{equation}\label{Gamma}
\omega=\sum_{j=1}^n \Gamma_j \delta_{z_j}, \quad \sum_{j=1}^n\Gamma_j =0,
\end{equation}
where $\Gamma_j\in\R$ is the strength of the vortex at $z_j$.

\item We have prescribed periods $a_k$, $b_k$ for the velocity $1$-form $\nu$ around a canonical homology basis, $\alpha_k$, $\beta_k$, which we take to be fixed
curves avoiding the vortex points for the time interval under consideration (taken to be short enough).

\end{itemize}

Under these assumptions
\begin{equation}\label{nuphipsi}
\nu +\I *\nu=d\varphi +\I d\psi =d\Phi.
\end{equation}
is a meromorphic differential on $M$ determined by the data
$$
\begin{cases}
\nu +\I *\nu = \sum \frac{\Gamma_j}{2\pi i} \frac{dz}{z-z_j} + \text{ regular analytic} \\
\oint_{\alpha_k} {\nu} = a_k, \\
\oint_{\beta_k} {\nu} = b_k.
\end{cases}
$$
In other words, $\nu +\I *\nu =d\Phi$ is an Abelian differential of the third kind on $M$, 
and it starts  evolving in time once $M$ has been provided with a Riemannian metric. 

Similarly, $\Phi$  is an Abelian integral (additively multi-valued in general) on $M$. Near a point vortex $z_k=z_k(t)$ of strength $\Gamma_k$ we expand, like in Lemma~\ref{lem:connections},
\begin{equation}\label{Phic}
\Phi=\varphi+\I\psi=\frac{\Gamma_k}{2\pi \I} \left(\log (z-z_k)- c_0-c_1(z-z_k)-\dots\right),
\end{equation}
where the coefficients $c_j=c_j(z_k)$ actually depend on all data, and also on the choice of the local coordinate at the vortex point. 
By Lemma~\ref{lem:connections} the first few coefficients $c_j$ behave as connections
under changes of coordinates for $z_k$.

The Green function
\begin{equation}\label{GGammaG}
G^\omega(z) = \sum_{k=1}^n \Gamma_k G(z,z_k)
\end{equation}
takes care of the singularities, but not of the periods, when considered as part of the stream function $\psi$. 
One need an extra term 
to absorb the difference between $\psi$ and $G^\omega$, and this is the harmonic form $\eta$ in (\ref{nudG}).
Since this $\eta$ is locally exact it can be written $\eta=dU$ with $U$ a harmonic function (multivalued if considered on all $M$). Then
\begin{equation}\label{psiGu}
\psi=G^\omega + U^*,
\end{equation}
modulo additive local constants, and with $U^*$ denoting the harmonic conjugate of $U$ (this notation is compatible with the Hodge star:
$d(U^*)=*dU$). Since $G^\omega$ is single-valued one can think of $U^*$ as representing the multi-valuedness of $\psi$.

The Green potential $G^\omega$ is harmonic away from the vortex points, while each individual $G(\cdot,z_k)=G^{\delta_{z_k}}$ appearing in (\ref{GGammaG}) is not, 
because its Laplacian contains compensating volume terms. However, the $G(\cdot,z_k)$ are good enough to single out individual Robin functions at the vortices. 
Recalling (\ref{GlogH}), (\ref{greentaylor}) we have
$$
G(z,z_k)= \frac{1}{2\pi}\re\Big(- \log (z-z_k) + h_0(z_k) +h_1(z_k)(z-z_k)+
$$
$$
+h_{2}(z_k)(z-z_k)^2+h_{11}(z_k)(z-z_k)(\bar{z}-\bar{z_k})+ \mathcal{O}(|z-z_k|^3\Big).
$$
For the full stream function we have similarly, by (\ref{Phic}), (\ref{GGammaG}), (\ref{psiGu}),
$$
\psi(z)= \frac{\Gamma_k}{2\pi}\re\big(- \log (z-z_k) + c_0(z_k) +c_1(z_k)(z-z_k)+\dots\big)=
$$
$$
=\Gamma_k G(z,z_k)+\sum_{j\ne k}\Gamma_j G(z,z_j)+U^*(z).
$$
The constant term $c_0(z_k)$ in $\psi$ can therefore be identified as
\begin{equation}\label{cGammah}
c_0(z_k)=h_0(z_k) +\sum_{j\ne k} \frac{2\pi\Gamma_j}{\Gamma_k} \,G(z_k,z_j)+\frac{2\pi}{\Gamma_k}U^*(z_k).
\end{equation}
Similarly for $c_1(z_k)$:
\begin{equation}\label{cGammah1}
c_1(z_k)= h_1(z_k) +\sum_{j\ne k} \frac{4\pi\Gamma_j}{\Gamma_k}\, \frac{\partial G(z_k,z_j)}{\partial z_k}+\frac{4\pi}{\Gamma_k}\frac{\partial U^*(z_k)}{\partial z_k}.
\end{equation}

On using (\ref{Phic}) the analytic completion (\ref{nuphipsi}) of $\nu$ becomes 
\begin{equation*}\label{Phinu}
\nu +\I *\nu =  d\Phi=\frac{\Gamma_k}{2\pi i} \left(\frac{1}{z-z_k} - c_1(z_k) -2 c_2(z_k)(z-z_k)+\dots\right)dz.
\end{equation*}
This gives the expansions
\begin{equation*}\label{fdotPhi}
f(z,t)=-\dot{\Phi}(z,t)=-\frac{\Gamma_k}{2\pi\I}  (\frac{\dot{z}_k}{z-z_k} +\mathcal {O}(1)),
\end{equation*}
\begin{equation*}\label{dpsidz}
2\frac{\partial \psi}{\partial z}=-\I \Phi'=\frac{\Gamma_k}{2\pi}  (-\frac{1}{z-z_k}+c_1(z_k)+\mathcal{O}(z-z_k)),
\end{equation*}
\begin{equation*}\label{d2psidz}
4(\frac{\partial \psi}{\partial z})^2=-(\Phi')^2=\frac{\Gamma_k^2}{4\pi^2}  \left(\frac{1}{(z-z_k)^2}-\frac{2c_1(z_k)}{z-z_k}+\mathcal {O}(1)\right).
\end{equation*}
Now it is only to insert all this into (\ref{eulerweak1}), with $a=z_k$ and $D$ chosen such that $z_k$ is the only vortex inside $D$.
The left member is evaluated by straight-forward residue computations, and the integrand in right member vanishes because $\partial\psi/\partial z$
is analytic in $D\setminus \D(z_k,\varepsilon)$. Out of this comes the vortex dynamics right away:

\begin{theorem}\label{thm:dzkdt}
The velocity of the vortex $z_k(t)$ is given by
\begin{align}\label{dzkdt}
\lambda(z_k)^2\frac{d {z}_k}{d t} &= \frac{\Gamma_k}{2\pi\I }\left(\overline{c_1(z_k)}+\frac{\partial}{\partial \bar{z}_k}\log \lambda(z_k)\right)\\
&=\frac{\Gamma_k}{4\pi i}\left(\overline{r_{\rm metric}(z_k)}-\overline{r_{\rm robin}(z_k)}\right),
\end{align}
where $c_1(z_k)$ is given by (\ref{cGammah1}) and where we for the second equation have set
$$
\begin{cases}
r_{\rm metric}(z)&=2\frac{\partial}{\partial z}\log \lambda(z),\\
r_{\rm robin}(z)&=-2c_1(z).
\end{cases}
$$
\end{theorem}

It remains to clarify in more detail how $r_{\rm robin}$ depends on the data $z_j$, $\Gamma_j$, $a_j$, $b_j$,
namely on how the function $U$ in (\ref{cGammah1}) depends on these data.
This missing information will be provided in the next section, more precisely by equations (\ref{U}), (\ref{Ak}), (\ref{Bk}) there. 

We see from the theorem that the motion of a vortex is due to two different sources: the presence of other vortices and circulating flows, summarized in $r_{\rm robin}$,
and the variation of the metric, summarized in $r_{\rm metric}$. Each of these is an {affine connection}
in the sense Section~\ref{sec:connections}. The difference between two affine connections is a differential, or covariant vector.
This makes the equation consistent from a tensor analysis point of view since $dz_k/dt$ is to be considered as a contravariant vector and the factor 
$\lambda^2=g_{11}=g_{22}$ (with $g_{12}=g_{21}=0$) transforms it into a covariant vector.


\section{Energy renormalization and the Hamiltonian}\label{sec:hamiltonian}

The {\it total (kinetic) energy} ${\mathcal{E}}(\nu,\nu)$ of the flow  $\nu=d\varphi=-*d\psi$ is 
in the smooth case given by
$$
2{\mathcal{E}}(\nu,\nu)=\int_M \nu\wedge *\nu=\frac{\I}{2}\int_M (\nu+\I*\nu)\wedge (\nu-\I*\nu)
$$
$$
=\frac{\I}{2}\int_M d\Phi\wedge d\bar{\Phi}=\int_M d\psi\wedge *d\psi=\int_M d\varphi\wedge *d\varphi
$$
$$
=\int_M (*dG^\omega +\eta)\wedge (-dG^\omega +*\eta)
$$
$$
=\int_M dG^\omega \wedge *dG^\omega + \int_M \eta\wedge*\eta
=\mathcal{E}(\omega,\omega)+ \int_M \eta\wedge*\eta
$$
$$
=\int_M G^\omega\wedge \omega +\sum_{j=1}^{\texttt g}(\oint_{\alpha_j}\eta\oint_{\beta_j}*\eta-\oint_{\alpha_j}*\eta\oint_{\beta_j}\eta).
$$
The last term (the sum) is obtained in a standard fashion by cutting the surface along the curves $\alpha_j$, $\beta_j$ to obtain a plane
polygon with $4\texttt{g}$ sides which are pairwise identified with each other. For closed differentials $\theta$, $\tilde{\theta}$ in general this gives
\begin{equation}\label{bilinear}
\int_M \theta \wedge \tilde{\theta}=\sum_{j=1}^\texttt{g} \big( \oint_{\alpha_j}\theta \oint_{\beta_j}\tilde{\theta}-  \oint_{\alpha_j}\tilde{\theta} \oint_{\beta_j}{\theta} \big),
\end{equation}
see for example \cite{Farkas-Kra-1992}.

In the point vortex limit the above expression for the energy becomes, formally,
$$
2{\mathcal{E}}(\nu,\nu)=\sum_{k,j=1}^n \Gamma_k \Gamma_jG(z_k,z_j) +\sum_{j=1}^{\texttt g}\big(\oint_{\alpha_j}\eta\oint_{\beta_j}*\eta-\oint_{\alpha_j}*\eta\oint_{\beta_j}\eta\big),
$$
but the presence of the terms $\Gamma_k^2G(z_k,z_k)$ makes the first term become infinite.
Looking at the expansion (\ref{GlogH}), (\ref{greentaylor}) of the Green's function 
suggests that the infinities could be subtracted away simply by omitting the logarithmic terms, i.e. by replacing $G(z_k,z_k)$ by (minus) $\frac{1}{2\pi}h_0(z_k)$.
However this turns out to be too crude, for example it makes the energy become a $0$-connection instead of a function (as depending on $z_1,\dots,z_n$).
This particular problem can be resolved by adding $\log \lambda (z_k)$ to $h_0(z_k)$, i.e., by subtracting the $0$-connection which is 
associated to the metric. 

It turns out that the above actually is the right thing to do, it gives the same equations of motion as more direct approaches give (e.g., that summarized in 
Theorem~\ref{thm:dzkdt}), or those available in the literature, like \cite{Hally-1980, Dritschel-Boatto-2015}.
Thus we introduce the {\it renormalized energy}, or {\it Hamiltonian function} $H(z_1,\dots,z_n)$, by
$$
2H(z_1,\dots,z_n)=2{\mathcal{E}}_{\rm renorm}(\nu,\nu)=\frac{1}{2\pi}\sum_{k=1}^n \Gamma_k^2 (h_0(z_k)+\log \lambda(z_k))+
$$
$$
+\sum_{k\ne j} \Gamma_k \Gamma_jG(z_k,z_j) +\sum_{j=1}^{\texttt g}\big(\oint_{\alpha_j}\eta\oint_{\beta_j}*\eta-\oint_{\alpha_j}*\eta\oint_{\beta_j}\eta\big).
$$
Here also the final terms depend, although somewhat more indirectly, on $z_1, \dots, z_n$, and next we proceed to elaborate this dependence. 

Since $\sum_{j=1}^\texttt{g}\Gamma_j=0$, the formal $0$-chain  $ \sum_{j=1}^\texttt{g} \Gamma_j\, ({z_j})$
is a boundary, i.e., there exists a $1$-chain $\gamma$ with real coefficients such that, in the sense of homology,
\begin{equation}\label{sigmaGamma}
\partial\gamma= \sum_{k=1}^\texttt{g} \Gamma_k \,({z_k}).
\end{equation}
This is an equality between $0$-chains, and the sum is just formal.
Explicitly, $\gamma$ is a linear combination, with real coefficients, of oriented arcs such that (\ref{sigmaGamma}) holds.
One may for example take $\gamma= \sum_{j=1}^n \Gamma_j \,\gamma_j$, where each $\gamma_j$ is an arc from some base point $b$ to $z_j$.
These (time-dependent) arcs can be chosen so that they, under a small time interval under consideration, do not intersect
the homology basis $\{\alpha_k, \beta_k\}$. 
 
Using $\gamma$ the periods of $*dG^\omega$ can be exhibited more clearly as 
\begin{equation}\label{periodsstardGa}
\oint_{\alpha_k}*dG^\omega =\int_\gamma dU_{\alpha_k},
\end{equation}
\begin{equation}\label{periodsstardGb}
\oint_{\beta_k}*dG^\omega =\int_\gamma dU_{\beta_k}.
\end{equation}
The relations (\ref{periodsstardGa}), (\ref{periodsstardGa}) generalize the corresponding vortex pair equations (\ref{ointdVa}), (\ref{ointdVb}),
in fact they can be obtained by taking linear combinations of relations of the latter form. 

The differentials $dU_{\alpha_k}$, $dU_{\beta_k}$  are fixed (independent of $z_1,\dots, z_n$), so 
all dependence on vortex locations now lies in the $1$-chain $\gamma$:
$$
\gamma =\gamma(z_1, \dots,z_n).
$$
Further on, we may use (\ref{sigmaGamma}) to express (\ref{periodsstardGa}), (\ref{periodsstardGb}) 
in terms of the integrals $dU_{\alpha_k}$, $dU_{\beta_k}$ as in (\ref{Phioint}), (\ref{Uab}):
\begin{equation}\label{periodsstardGc}
\oint_{\alpha_k}*dG^\omega = \int_\gamma dU_{\alpha_k}= \sum_{j=1}^n\Gamma_j U_{\alpha_k}(z_j,b),
\end{equation}
\begin{equation}\label{periodsstardGd}
\oint_{\beta_k}*dG^\omega = \int_\gamma dU_{\beta_k}= \sum_{j=1}^n\Gamma_j U_{\beta_k}(z_j,b).
\end{equation}
Here $b$ may be taken to be the base point mentioned above (after (\ref{sigmaGamma})), but actually the right members 
in (\ref{periodsstardGc}), (\ref{periodsstardGd}) do not depend on $b$.
The functions $U_{\alpha_k}$, $U_{\beta_k}$ have unlimited harmonic continuations all over $M$, then becoming additively multivalued,
but the assumption $\sum \Gamma_j=0$ guarantees that  
whatever branch one chooses on ${\rm supp\,}\gamma$, the final result in the right members above will always be the same.

Now $\eta$ being a harmonic $1$-form it can be expanded as
\begin{equation}\label{etaAB}
\eta=-\sum_{j=1}^\texttt{g} A_j dU_{\beta_j}+\sum_{j=1}^\texttt{g} B_j dU_{\alpha_j},
\end{equation}
where $A_j=A_j(z_1,\dots, z_n)$, $B_j=B_j(z_1,\dots,z_n)$ are coefficients which depend on the locations of the poles. This can also be written (locally) as
$$
\eta=d\big( \sum_{j=1}^\texttt{g} (-A_j U_{\beta_j}+ B_j U_{\alpha_j})\big)
=-*d\big( \sum_{j=1}^\texttt{g} (-A_j U_{\beta_j}^*+ B_j U_{\alpha_j}^*)\big)=-*d(U^*),
$$
by which we have identified the function $U$ appearing in (\ref{psiGu}) as
\begin{equation}\label{U}
U= \sum_{j=1}^\texttt{g} (-A_j U_{\beta_j}+ B_j U_{\alpha_j}).
\end{equation}

The coefficients $A_k$, $B_k$ also depend on the prescribed circulations $a_k$, $b_k$ of the flow,
so that, using (\ref{taualpha}), (\ref{taubeta}), (\ref{periodsstardGc}), (\ref{periodsstardGd}), (\ref{etaAB}),
$$
a_k=\oint_{\alpha_k}\eta-\oint_{\alpha_k}*dG^\omega =A_k-\int_\gamma dU_{\alpha_k}=A_k-\sum_{j=1}^n \Gamma_j U_{\alpha_k}(z_j),
$$
$$
b_k=\oint_{\beta_k}\eta -\oint_{\beta_k}*dG^\omega =B_k- \int_\gamma dU_{\beta_k}=B_k-\sum_{j=1}^n \Gamma_j U_{\beta_k}(z_j).
$$
Thus
\begin{equation}\label{Ak}
A_k(z_1,\dots, z_n)=a_k+\sum_{j=1}^n \Gamma_j U_{\alpha_k}(z_j),
\end{equation}
\begin{equation}\label{Bk}
B_k(z_1,\dots,z_n)=b_k+\sum_{j=1}^n \Gamma_j U_{\beta_k}(z_j).
\end{equation}
showing how $\eta$ depends, via (\ref{etaAB}), on the parameters $z_1,\dots, z_n$. 
Taking the derivative with respect to $z_1$ (for example) gives
$$
\frac{\partial A_k}{\partial z_1}= \Gamma_1\,\frac{\partial U_{\alpha_k}(z_1)}{\partial z_1},
$$
$$
\frac{\partial B_k}{\partial z_1}= \Gamma_1\,\frac{\partial U_{\beta_k}(z_1)}{\partial z_1}.
$$
This will be used in the Section~\ref{sec:hamiltoneq}.

Now using (\ref{etaAB}), (\ref{U}), (\ref{Ak}), (\ref{Bk}) we can compute more explicitly the energy contribution from $\eta$.
We first notice that, by (\ref{bilinear}),
$$
\int_M dU_{\beta_k}\wedge*dU_{\beta_j}=-\oint_{\beta_k}*dU_{\beta_j}, \quad
 \int_M dU_{\beta_k}\wedge*dU_{\alpha_j}=-\oint_{\beta_k}*dU_{\alpha_j},
$$
$$
\int_M dU_{\alpha_k}\wedge*dU_{\beta_j}=-\oint_{\alpha_k}*dU_{\beta_j}, \quad
 \int_M dU_{\alpha_k}\wedge*dU_{\alpha_j}=-\oint_{\alpha_k}*dU_{\alpha_j}.
$$
Since $(\omega_1,\omega_2)_1=\int_M \omega_1\wedge *\omega_2$ is an inner product on $1$-forms, in particular on the $2\texttt{g}$
dimensional space of harmonic forms, for which we have used $(-dU_{\beta_k})_{k=1}^\texttt{g}$, $(dU_{\alpha_k})_{k=1}^\texttt{g}$ 
as a basis in (\ref{etaAB}) (see also (\ref{taualpha}), (\ref{taubeta})), it follows that the matrix (written here in block form)
$$
\left( \begin{array}{cc}
                                            ( \int_M dU_{\beta_k}\wedge*dU_{\beta_j}) & (- \int_M dU_{\beta_k}\wedge*dU_{\alpha_j})  \\
                                              (-\int_M dU_{\alpha_k}\wedge*dU_{\beta_j}) & (  \int_M dU_{\alpha_k}\wedge*dU_{\alpha_j})  \\
\end{array} \right)=
$$
\begin{equation}\label{ointdU}
=
\left( \begin{array}{cc}
                                             (- \oint_{\beta_k}*dU_{\beta_j}) & ( \oint_{\beta_k}*dU_{\alpha_j})  \\
                                              ( \oint_{\alpha_k}*dU_{\beta_j}) & ( -\oint_{\alpha_k}*dU_{\alpha_j})  \\
\end{array} \right)
\end{equation}
is symmetric and positive definite. 
Using  (\ref{etaAB}) we can then write the energy of $\eta$ as a quadratic form:
$$
\int_M \eta \wedge *\eta= \sum_{k=1}^{\texttt g}(\oint_{\alpha_k}\eta\oint_{\beta_k}*\eta-\oint_{\beta_k}\eta\oint_{\alpha_k}*\eta)=
$$
$$ 
+
\left(
 \begin{array}{cc}
(A_k),&(B_k)
\end{array}\right)
\left( \begin{array}{cc}
                                             (- \oint_{\beta_k}*dU_{\beta_j}) & ( \oint_{\beta_k}*dU_{\alpha_j})  \\
                                              ( \oint_{\alpha_k}*dU_{\beta_j}) & ( -\oint_{\alpha_k}*dU_{\alpha_j})  \\
\end{array} \right)
\left(\begin{array}{c}
(A_j)\\
(B_j)
\end{array}\right).
$$

Summarizing for the full Hamiltonian we have:
\begin{equation}\label{H}
2H(z_1,\dots,z_n)=
\end{equation}
$$ 
\left( \begin{array}{cccc}
\Gamma_1,&\Gamma_2& \ldots &\Gamma_n
\end{array}\right)
\left( \begin{array}{cccc}
                                             R(z_1)& G(z_1,z_2)&\ldots &G(z_1,z_n)  \\
                                              G(z_2,z_1) & R(z_2)&\ldots &G(z_2,z_n) \\
\vdots& \vdots & \ddots &\vdots\\
G(z_n,z_1)& G(z_n, z_2)&\dots &R(z_n)
\end{array} \right)
\left(\begin{array}{c}
\Gamma_1\\
\Gamma_2\\
\vdots\\
\Gamma_n
\end{array}\right)+
$$
$$ 
+
\left( \begin{array}{cc}
(A_k),&(B_k)
\end{array}\right)
\left( \begin{array}{cc}
                                             (- \oint_{\beta_k}*dU_{\beta_j}) & ( \oint_{\beta_k}*dU_{\alpha_j})  \\
                                              ( \oint_{\alpha_k}*dU_{\beta_j}) & ( -\oint_{\alpha_k}*dU_{\alpha_j})  \\
\end{array} \right)
\left(\begin{array}{c}
(A_j)\\
(B_j)
\end{array}\right),
$$
where we have set
$$
R(z)= \frac{1}{2\pi} (h_0(z)+\log\lambda(z))
$$
for what is the same as the Robin function defined in terms of invariant distance. Clearly the first term in $H$ agrees with the Hamiltonian
given in \cite{Dritschel-Boatto-2015}.


\section{Hamilton's equations}\label{sec:hamiltoneq}

The equations of motion for the vortices should be obtained by differentiating $H(z_1,\dots,z_n)$ with respect to the $\bar{z}_k$.
More accurately, the formulation of Hamilton's equation requires, besides the Hamiltonian function itself, the introduction of a phase 
space provided with a symplectic form. The phase space consists in our case of all possible configurations of the vortices,
collisions not allowed, i.e., we take it to be
$$
P=\{(z_1,\dots,z_n): z_j\in M, z_k\ne z_j \text{ for } k\ne j\},
$$
and the symplectic form on $P$ is to be
$$
\Omega=\sum_{k=1}^n \Gamma_k{\rm vol}^2(z_k)=-\frac{1}{2 \I}\sum_{k=1}^n \Gamma_k\lambda(z_k)^2 dz_k\wedge d\bar{z}_k.
$$
These are standard choices, known from \cite{Hally-1980, Boatto-Koiller-2008, Boatto-Koiller-2013}, for example.

The Hamilton equations are, in general terms (see \cite{Arnold-1978, Arnold-Khesin-1998, Frankel-2012}) and with a for us suitable
choice of sign,
$$
i(\frac{d}{dt}(z_1,\dots,z_n))\Omega= dH.
$$
This spells out, in our case, to
\begin{equation}\label{hamilton}
\Gamma_k\lambda(z_k)^2 \frac{dz_k}{dt}=-2 \I \frac{H(z_1,\dots, z_n)}{\partial \bar{z}_k}.
\end{equation}

\begin{theorem}\label{thm:hamiltoneq}
The Hamiltonian equations (\ref{hamilton}) agree with the vortex dynamics
given in Theorem~\ref{thm:dzkdt}.
\end{theorem}

\begin{proof}
Choosing $k=1$, for example, in (\ref{hamilton}) gives, on using (\ref{c1c0}) together with the symmetries of $G(z_k,z_j)$ and of the 
quadratic form representing the energy of $\eta$,
$$
-2\I\frac{\partial H(z_1,\dots,z_n)}{\partial  \bar{z}_1}=
$$
$$
=\frac{\Gamma_1^2}{2\pi\I}\big(\overline{h_1(z_1)}+\frac{\partial}{\partial \bar{z}_1}\log\lambda (z_1)\big)
-2\I\Gamma_1\sum_{j=2}^n \Gamma_j\frac{\partial G(z_1,z_j)}{\partial \bar{z}_1}+
$$
$$
+2\I\sum_{k,j=1}^\texttt{g}A_k\frac{\partial A_j}{\partial \bar{z}_1}  \oint_{\beta_k}  *dU_{\beta_j}
-2\I\sum_{k,j=1}^\texttt{g}A_k\frac{\partial B_j}{\partial \bar{z}_1}  \oint_{\beta_k}  *dU_{\alpha_j}-
$$
$$
-2\I\sum_{k,j=1}^\texttt{g}B_k\frac{\partial A_j}{\partial \bar{z}_1} \oint_{\alpha_k}  *dU_{\beta_j}
+2\I\sum_{k,j=1}^\texttt{g}B_k\frac{\partial B_j}{\partial \bar{z}_1} \oint_{\alpha_k}  *dU_{\alpha_j}= 
$$
$$
=\frac{\Gamma_1^2}{2\pi\I}\big(\overline{h_1(z_1)}+\frac{\partial}{\partial \bar{z}_1}\log\lambda (z_1)\big)
-2\I\Gamma_1\sum_{j=2}^n \Gamma_j\frac{\partial G(z_1,z_j)}{\partial \bar{z}_1}+
$$
$$
+2\I\Gamma_1\sum_{k,j=1}^\texttt{g}A_k\frac{\partial U_{\alpha_j}(z_1)}{\partial \bar{z}_1}
 \oint_{\beta_k}  *dU_{\beta_j}
-2\I\Gamma_1\sum_{k,j=1}^\texttt{g}A_k\frac{\partial U_{\beta_j}(z_1)}{\partial \bar{z}_1}
 \oint_{\beta_k}  *dU_{\alpha_j}-
$$
$$
-2\I\Gamma_1\sum_{k,j=1}^\texttt{g}B_k\frac{\partial U_{\alpha_j}(z_1)}{\partial \bar{z}_1}
 \oint_{\alpha_k}  *dU_{\beta_j}
+2\I\Gamma_1\sum_{k,j=1}^\texttt{g}B_k\frac{\partial U_{\beta_j}(z_1)}{\partial \bar{z}_1}
 \oint_{\alpha_k}  *dU_{\alpha_j}.
$$

This is to be compared with what comes out from Theorem~\ref{thm:dzkdt} combined with (\ref{cGammah1}), (\ref{U}), namely
$$
\Gamma_1\lambda(z_1)^2 \frac{dz_1}{dt}=
$$
$$
=\frac{\Gamma_1^2}{2\pi \I}\big(\overline{h_1(z_k)}+\frac{\partial}{\partial \bar{z_1}}\log \lambda (z_1)
+\sum_{j=2}^n \frac{4\pi \Gamma_j}{\Gamma_1}\frac{\partial G(z_1,z_j)}{\partial \bar{z}_1}
+\frac{4\pi}{\Gamma_1}\frac{\partial U^*(z_1)}{\partial \bar{z_1}}\big)
$$
$$
=\frac{\Gamma_1^2}{2\pi \I}(\overline{h_1(z_k)}+\frac{\partial}{\partial \bar{z_1}}\log \lambda (z_1))-
$$
$$
-2\I\sum_{j=2}^n {\Gamma_1\Gamma_j}\frac{\partial G(z_1,z_j)}{\partial \bar{z}_1}
+2{\I}\Gamma_1\sum_{k=1}^\texttt{g} A_k\frac{\partial U^*_{\beta_k}(z_1)}{\partial \bar{z}_1}
-2\I\Gamma_1\sum_{k=1}^\texttt{g}B_k\frac{\partial U^*_{\alpha_k}(z_1)}{\partial \bar{z}_1}.
$$

We see that we have agreement between the two versions of dynamics provided just the terms containing $A_k$ and $B_k$ match
individually, i.e. (choosing $A_k$ for example)
$$
\sum_{k,j=1}^\texttt{g}\frac{\partial U_{\alpha_j}(z_1)}{\partial \bar{z}_1}
 \oint_{\beta_k}  *dU_{\beta_j}
-\sum_{k,j=1}^\texttt{g}\frac{\partial U_{\beta_j}(z_1)}{\partial \bar{z}_1}
 \oint_{\beta_k}  *dU_{\alpha_j}
=\sum_{k=1}^\texttt{g} \frac{\partial U^*_{\beta_k}(z_1)}{\partial \bar{z}_1}.
$$
This is a special case of
$$
\sum_{j=1}^\texttt{g}d U_{\alpha_j}
 \oint_{\beta_k}  *dU_{\beta_j}
-\sum_{j=1}^\texttt{g}d U_{\beta_j}
 \oint_{\beta_k}  *dU_{\alpha_j}
= *dU_{\beta_k}.
$$
Such a linear relation between harmonic forms $dU_{\alpha_j}$, $dU_{\beta_j}$, $*dU_{\beta_k}$ holds if and only if
both members have the same integrals around all curves $\alpha_\ell$, $\beta_\ell$ in the homology basis. 
And this follows on using the period properties (\ref{taualpha}), (\ref{taubeta}) of the differentials above 
together with the symmetry of the matrix (\ref{ointdU}). This finishes the proof.
\end{proof}


\bibliography{bibliography_gbjorn.bib}

\end{document}